\documentclass[pra,twocolumn,superscriptaddress,showpacs]{revtex4-1}

\usepackage[commandnameprefix=always]{changes}
\usepackage{latexsym}
\usepackage{graphicx}
\usepackage[colorlinks=true, citecolor=blue, urlcolor=blue]{hyperref}
\usepackage{float}
\usepackage{textcomp}
\usepackage{mathpazo}
\usepackage{comment}
\usepackage{lipsum}

\usepackage{exscale}
\usepackage{bbm}
\usepackage{latexsym}
\usepackage[T1]{fontenc}
\usepackage{enumerate}
\usepackage{bbold}
\usepackage{color}
\usepackage[colorlinks=true,citecolor=blue,urlcolor=blue]{hyperref}
\usepackage{mathtools}
\usepackage{amsfonts,amsmath,amssymb,amsthm}

\usepackage[normalem]{ulem}

\usepackage{xr}
\makeatletter
\newcommand*{\addFileDependency}[1]{
  \typeout{(#1)}
  \@addtofilelist{#1}
  \IfFileExists{#1}{}{\typeout{No file #1.}}
}
\makeatother

\newcommand*{\myexternaldocument}[1]{
    \externaldocument{#1}
    \addFileDependency{#1.tex}
    \addFileDependency{#1.aux}
}

\myexternaldocument{appendix}

\usepackage{bbm}

\usepackage[colorlinks=true,citecolor=blue,urlcolor=blue]{hyperref}
\usepackage[draft]{fixme}
\usepackage{amsmath,bbm}
\usepackage{graphicx}
\usepackage{amsfonts}
\usepackage{amssymb}
\usepackage{amsmath, amssymb, amsthm,verbatim,graphicx,bbm}
\usepackage{mathrsfs}
\usepackage{color,xcolor,longtable}
\usepackage{changes}

\newcommand{\beq}[0]{\begin{equation}}
\newcommand{\eeq}[0]{\end{equation}}

\newcommand{\one}{\leavevmode\hbox{\small1\normalsize\kern-.33em1}}

\def\be{\begin{equation}}
\def\ee{\end{equation}}
\def\ben{\begin{eqnarray}}
\def\een{\end{eqnarray}}
\def\eea{\end{array}}
\def\bea{\begin{array}}

\newcommand{\Tr}[1]{\mathrm{Tr}#1}
\newcommand{\bei}{\begin{itemize}}
\newcommand{\eei}{\end{itemize}}
\newcommand{\ket}[1]{\left|#1\right\rangle}
\newcommand{\bra}[1]{\left\langle#1\right|}

\newcommand{\proj}[1]{\ket{#1}\!\bra{#1}}

\newcommand{\I}{\mathbbm{1}}

\newcommand{\im}{\mathbbm{i}}

\renewcommand{\emph}[1]{\textbf{#1}}

\makeatletter
\newtheorem*{rep@theorem}{\rep@title}
\newcommand{\newreptheorem}[2]{%
\newenvironment{rep#1}[1]{%
 \def\rep@title{#2 \ref{##1}}%
 \begin{rep@theorem}}%
 {\end{rep@theorem}}}
\makeatother

\theoremstyle{plain}
\newtheorem{thm}{Theorem}
\newtheorem*{thm*}{Theorem}
\newreptheorem{thm}{Theorem}

\newtheorem{defn}{Definition}

\theoremstyle{definition}

\theoremstyle{remark}

\usepackage[T1]{fontenc}

\begin{document}
\title{Certification of the maximally entangled state using non-projective measurements}

\author{Shubhayan Sarkar}
\email{sarkar@cft.edu.pl}
\affiliation{Center for Theoretical Physics, Polish Academy of Sciences, Aleja Lotnik\'{o}w 32/46, 02-668 Warsaw, Poland}

\begin{abstract}	
In recent times, device-independent certification of quantum states has been one of the intensively studied areas in quantum information. However, all such schemes utilise projective measurements which are practically difficult to generate.
In this work, we consider the one-sided device-independent (1SDI) scenario and propose a self-testing scheme for the two-qubit maximally entangled state using  non-projective measurements, in particular, three three-outcome extremal POVM's. We also analyse the robustness of our scheme against white noise.
\end{abstract}

\maketitle
\textit{Introduction.---} The existence of nonlocal correlations, as was first realised by Einstein, Podolski and Rosen in 1935 \cite{EPR} as a paradox and then subsequently by Schr$\mathrm{\ddot{o}}$dinger in the same year \cite{Schrod}, is one of the most intriguing features of quantum theory. Consequently, Bell \cite{Bell, Bell66} proposed a mathematical formulation to detect whether quantum theory is inherently nonlocal or not, and thus it is commonly referred to as Bell nonlocality. Apart from its relevance in the foundations of physics, Bell nonlocality has given rise to an enormous number of applications in computation, communication and information theory \cite{NonlocalityReview}. 


A recent well profound application of nonlocality is device-independent (DI) certification where assuming quantum theory and some other physically well-motivated assumptions, the statistics obtained from a black box are enough to validate the underlying mechanism inside it. The strongest DI certification is referred to as self-testing. First introduced in \cite{Yao}, self-testing allows one to certify the underlying quantum states and the measurements, upto some freedom based on the maximal violation of a Bell inequality \cite{SupicReview}. In recent times, there  has been increased interest to find protocols to self-test various quantum systems due to their applicability in various quantum information tasks. Despite the progress in designing schemes to self-test various quantum states using projective measurements, for instance, Refs. \cite{Scarani, Bamps, All, chainedBell, Projection,  sarkar, prakash, Jed1, sarkar4, sarkar5, Allst1}, there is no scheme that utilises non-projective measurements. 

The precision to experimentally generate sets of projective measurements, which are pre-requisite to observe any form of non-locality, reduces as the dimension of the system grows [for instance see \cite{exp1}]. Thus, a natural question arises whether noisy measurements or non-projective measurements can also be used to self-test quantum states. Further on to self-test any state or measurement, one needs to observe the maximal violation of an inequality which touches the set of quantum correlations, or simply the quantum set, at a particular point. It is also an open question in quantum foundations whether a point on the boundary of the quantum set in some scenarios can be saturated by only non-projective measurements. 

Self-testing quantum states using non-projective measurements is not possible in the standard Bell scenario. The reason being that the maximal violation of Bell inequalities 
can always be achieved by projective measurements. Consequently, we consider another form of nonlocality, known as quantum steering  \cite{Wiseman, Ecaval, Quin}. To witness quantum steering, one needs to consider the Bell scenario with an additional assumption that one of the parties is trusted. In the DI regime, this is referred to as one-sided device-independent (1SDI) scenario. Certification of quantum states and measurements in 1SDI scenario has gained recent interest \cite{Supic, Alex, Goswami, Bharti, Chen, sarkar6, sarkar3} as they are more robust to noise and require detectors with lower efficiencies when compared to fully DI scenarios \cite{Crypto1, Crypto2}.   

In this work, we provide the first scheme to certify the two-qubit maximally entangled state 
\begin{eqnarray}
\ket{\phi^+}=\frac{1}{\sqrt{2}}(\ket{00}+\ket{11})
\end{eqnarray}
using three three-outcome non-projective extremal measurements in the 1SDI scenario. For this purpose, we first construct a steering inequality with two parties such that each of them chooses three inputs and gets three outputs. We then use the maximal violation of this steering inequality to obtain the self-testing result. We finally show that our scheme is highly robust when the states and measurements are mixed with white noise.

\textit{Preliminaries---}Before proceeding to the results, let us first describe the scenario and notions used throughout this work. 

{\it{Extremal positive operator valued measure (POVM).}} Any measurement in quantum theory, usually referred to as a POVM, is represented as $M=\{M^a\}$ where $M^a$ are the measurement elements corresponding to the $a$-th outcome of $M$. These elements are positive semi-definite operators and $\sum_aM^a=\I$. Now, a POVM that can not be expressed as a convex combination of other POVM's is defined as an extremal POVM. As shown in \cite{peri} the elements $M^a$ of any rank-one extremal POVM can be expressed as $M^a=\lambda_a\proj{\nu_a}$ where $\lambda_a\geq 0$ and the elements are linearly independent.

\begin{figure}[t]
\includegraphics[width=\linewidth]{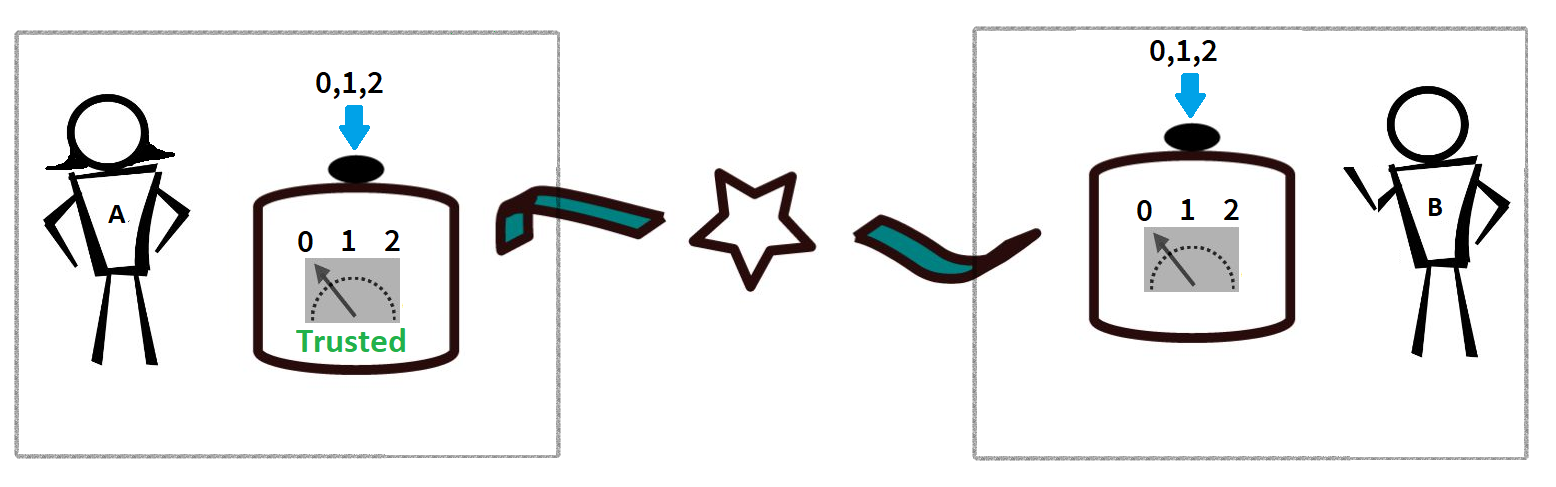}
    \caption{Quantum steering scenario. Alice and Bob are spatially separated and each of them receive a subsystem on which they perform three three-outcome measurements. They are not allowed to communicate during the experiment. Once it is complete, they construct the joint probability distribution $\{p(a,b|x,y)\}$.}
    \label{fig1}
\end{figure}

\textit{Quantum steering scenario.} In this work, we consider a simple scenario to witness quantum steering, consisting of two spatially separated parties namely, Alice and Bob. They locally perform measurements on their respective subsystems which they receive from a preparation device. Bob can choose among three measurements denoted by $B_y$ such that $y=0,1,2$ each of which results in three outcomes labelled by $b=0,1,2$. The measurement performed by Bob might affect the received subsystem with Alice which is denoted as $\sigma_b^y\in \mathcal{H}_A$ where $\sigma_b^y$ are positive semi-definite operators. The collection of these operators $\sigma=\{\sigma_b^y\ s.t.\ b=0,1,2,\ y=0,1,2\}$ is called an assemblage. 


In quantum theory the operators $\sigma_b^y$ are expressed for any $y,b$ as
\begin{eqnarray}
\sigma_b^y=\Tr_B\left[(\I_A\otimes N_y^b)\rho_{AB}\right]
\end{eqnarray}
where $\rho_{AB}\in\mathcal{H}_A\otimes\mathcal{H}_B$ is the state shared between Alice and Bob and $B_y=\{N_y^b\}$ denote Bob's measurements. 
Alice is trusted here, which means that her measurements are known or she can perform tomography on her subsystem. If the shared state is not steerable, then the assemblage has a local hidden state (LHS) model \cite{Wiseman}  defined as,

\begin{eqnarray}\label{LHS}
\sigma_b^y=\sum_\lambda p(\lambda) p_\lambda(b|y)\rho_\lambda
\end{eqnarray}
where $\sum_{\lambda}p(\lambda)=1$, $p_\lambda(b|y)$ are the probability distributions over $\lambda$ and $\rho_\lambda$ are density matrices over $\mathcal{H}_A$.
As Alice can perform topographically complete measurements on $\sigma^b_y$, in general, quantum steering is witnessed by so-called ``steering functional", a map from the assemblage $\{\sigma^b_y\}$ to a real number \footnote{A ``steering functional'' is defined by a set of Hermitian matrices $F^y_b$ that maps the assemblage to a real number by a functional of the form $\sum_{b,y}\Tr(F^b_y\sigma^b_y)$.}.

Instead of checking the steerability of the assemblage, quantum steering can be equivalently witnessed similar to a Bell scenario where trusted Alice and untrusted Bob performs the measurements $A_x=\{M^a_x\}$ and $B_y=\{N_y^b\}$ respectively and obtain the joint probability distribution $\vec{p}=\{p(a,b|x,y)\}$ where $a,b,x,y=0,1,2$. Here $a,x$ denotes the output and input of Alice respectively. The probabilities can be computed in quantum theory as,
\begin{eqnarray}
p(a,b|x,y)=\Tr\left[(M_x^a\otimes N_y^b)\rho_{AB}\right]=\Tr\left(M_x^a\sigma_b^y\right)
\end{eqnarray} 
To witness quantum steering, a steering inequality $\mathcal{B}$ can now be constructed from $\vec{p}$ as
\begin{eqnarray}
\mathcal{B}(\vec{p})=\sum_{a,b,x,y}c_{a,b|x,y}p(a,b|x,y) \leq\beta_L
\end{eqnarray}
where $c_{a,b|x,y}$ are real coefficients and $\beta_L$ denotes the maximum value attainable using assemblages admitting an LHS model \eqref{LHS}. The probabilities one obtains from such assemblages are expressed using \eqref{LHS} as
\begin{eqnarray}\label{LHS2}
p(a,b|x,y)=\sum_{\lambda} \ p(\lambda)p(a|x,\rho_\lambda)p(a|x,\lambda).
\end{eqnarray}
The above representation \eqref{LHS2} will be particularly useful to find the LHS bound of the steering inequality proposed in this work. In the DI framework, the above-presented scenario is also referred to as the 1SDI scenario [see Fig-\ref{fig1}].

{\it{Self-testing.}} Inspired by \cite{sarkar6}, we now define self-testing in the 1SDI scenario. 
\begin{defn}
Consider the above 1SDI scenario with the preparation device creating a state $\ket{\psi}_{AB}$. Alice and Bob perform measurements on this state and observe the joint probability distribution $\{p(a,b|x,y)\}$. Alice is trusted and her measurements $A_x$ are fixed and Bob's measurements are represented as $B_y=\{N^b_y\}$ are arbitrary. Let us now consider that the distribution $\{p(a,b|x,y)\}$ is generated by an ideal experiment with a state $\ket{\tilde{\psi}}_{AB}$ and Bob's measurements $\tilde{B}_y=\{\tilde{N}^b_y\}$. Then, the state $\ket{\psi}_{AB}$ and measurements $B_y$ are certified from $\{p(a,b|x,y)\}$ if there exists a unitary $U_B:\mathcal{H}_B\to \mathcal{H}_B$ such that 
\begin{equation}
 (\mathbbm{1}_A\otimes U_B)\ket{\psi}_{AB}=\ket{\tilde{\psi}}_{AB},
\end{equation}
and,
\begin{equation}
    U_B\,\Pi_B N^b_{y}\Pi_B\,U_B^{\dagger}=\tilde{N}^b_{y},
\end{equation}
where $\Pi_B$ is the projection onto the support of the local support $\rho_B=\Tr_A\left(\proj{\psi}_{AB}\right)$.
\end{defn}
\noindent Let us now proceed towards the results of this work.

{\it{Results.---}}
We begin by constructing a steering inequality stated using the joint probability distribution $\vec{p}$ as
\begin{eqnarray}\label{Stefun1}
W&=& \sum_{a,b,x=0}^2 p(a,b\ne a|x,y=x)\leq \beta_L.
\end{eqnarray}
Using the fact that $\sum_{a,b}p(a,b|x,y)=1$ for all $x,y$, we can simplify the above steering inequality as
\begin{eqnarray}\label{Stefun}
W&=& 3-\sum_{a,x=0}^2 p(a,a|x,x)\leq \beta_L.
\end{eqnarray}
Alice is trusted and performs the measurements $A_x=\{M_x^a\}_{a=0,1,2},$ where $x=0,1,2$ and the measurement elements are given as $M_x^a=\frac{2}{3}\proj{e_{a,x}}$. Here the vectors $\ket{e_{a,x}}\in \mathbbm{C}^2$ are given by 
\begin{eqnarray}\label{Aideamea}
&\ket{e_{0,0}}&=\ket{0},\quad\ \ \ \ \ \ \ \ \ \ \ \ \ \ \ \ \ \ \ \ \ \ \  \ket{e_{1,0}}=\frac{1}{2}\ket{0}+\frac{\sqrt{3}}{2}\ket{1},\nonumber\\&\ket{e_{2,0}}&=\frac{1}{2}\ket{0}-\frac{\sqrt{3}}{2}\ket{1},\quad\ \ 
\ket{e_{0,1}}=\ket{1},\nonumber\\ &\ket{e_{1,1}}&=\frac{\sqrt{3}}{2}\ket{0}+\frac{\im}{2}\ket{1},\ \ \quad \ket{e_{2,1}}=\frac{\sqrt{3}}{2}\ket{0}-\frac{\im}{2}\ket{1},\nonumber\\
&\ket{e_{0,2}}&=\frac{1}{\sqrt{2}}\left(\ket{0}+\im\ket{1}\right),\quad \ket{e_{1,2}}=\frac{1}{\sqrt{2}}\left(\ket{0}+e^{\frac{7\pi\im}{6}}\ket{1}\right),\nonumber\\ &\ket{e_{2,2}}&=\frac{1}{\sqrt{2}}\left(\ket{0}+e^{\frac{-\pi\im}{6}}\ket{1}\right).
\end{eqnarray}
Notice that Alice's measurement is extremal.

Let us now compute the LHS bound $\beta_L$ of the steering inequality \eqref{Stefun1}.
Using \eqref{LHS2}, we rewrite the steering functional $W$ from Eq. \eqref{Stefun} as
\begin{eqnarray}\label{steeopprob}
W = 3- \sum^{2}_{a,x=0}\sum_{\lambda} \ p(\lambda)p(a|x,\rho_\lambda)p(a|x,\lambda)
\end{eqnarray}
%
%
%
We focus on the last term in Eq. (\ref{steeopprob}) for each $x$ and notice that 
they can be bounded from below in the following way
\begin{eqnarray}
    \sum_x\sum_{a=0}^{2}\sum_{\lambda}&p(\lambda)&p(a|x,\rho_{\lambda})p(a|x,\lambda)\nonumber\\&\geq& \sum_x\sum_{\lambda}p(\lambda)\min_{a}\{p(a|x,\rho_{\lambda})\},
\end{eqnarray}
where $x=0,1,2$ and we used the fact that $\sum_ap(a|x,\lambda)=1$ for any $x$ and $\lambda$. Now, minimising over $\rho_{\lambda}$, we obtain
\begin{eqnarray}
\sum_x\sum_{\lambda}p(\lambda)\min_{a}\{p(a|x,\rho_{\lambda})\}\geq\ \ \ \ \ \ \ \qquad \nonumber\\
\sum_{\lambda}p(\lambda)\min_{\rho_\lambda}\sum_x\min_{a}\{p(a|x,\rho_\lambda)\}.
\end{eqnarray}
Using the fact $\sum_{\lambda}p(\lambda)=1$, we get that the LHS bound is upper bounded by
\begin{eqnarray}
\beta_L \leq 3-\min_{\ket\psi\in\mathbbm{C}^2}\sum_{x=0}^2\min_{a}\{p(a|x,\ket{\psi})\}.
\end{eqnarray}
Notice that since the steering functional $W$ is linear, the minimisation can be carried out over pure states. Numerically evaluating the above quantity by putting in Alice's measurements \eqref{Aideamea}, we find that 
\begin{eqnarray}\label{local}
\beta_L\leq 2.673
\end{eqnarray}
For this purpose, we choose a state $\ket{\psi}\in\mathbbm{C}^2$ parameterised using the Bloch representation as
\begin{eqnarray}
\ket{\psi}=\cos{\frac{\theta}{2}}\ket{0}+e^{\mathbbm{i}\phi}\sin{\frac{\theta}{2}}\ket{1}
\end{eqnarray}
where $0\leq\theta\leq\pi/2$ and $0\leq\phi\leq\pi$. Now, the probabilities $p(a|x,\ket{\psi})$ when Alice performs the measurement $M^a_x$ are given by $p(a|x,\ket{\psi})=|\bra{e_{a,x}}\psi\rangle|^2$ which are a function of $\theta,\phi$. Now, a simple optimisation over the parameters $\theta,\phi$ gives us the local bound \eqref{local}.

Let us now evaluate the quantum bound $\beta_Q$, the maximal value achievable using quantum states and measurements, of the steering functional $W$ \eqref{Stefun1}. The quantum bound is in fact the same as the algebraic bound of $W$ \eqref{Stefun1}, that is, $3$ and for instance can be achieved by $\ket{\psi}_{AB}=\ket{\phi^+}_{AB}$ and Bob's measurements $B_x=\{N^a_x\}_{a=0,1,2},$ such that $x=0,1,2$. Here, the measurement elements $N^a_x=\frac{2}{3}\proj{f_{a,x}}$ such that $\ket{f_{a,x}}=\ket{e_{a,x}^{*\perp}}\in\mathbbm{C}^2$, that is, $\langle e_{a,x}^{*}\ket{f_{a,x}}=0$ where $\ket{e_{a,x}}$ are specified in Eq. \eqref{Aideamea} and $*$ denotes their conjugate.

It is important to note here that to achieve the maximal violation of the steering inequality \eqref{Stefun1}, all the probabilities $p(a,a|x,x)$ in \eqref{Stefun} has to be $0$, that is, 
\begin{eqnarray}\label{SOS}
p(a,a|x,x)=0\quad a,x=0,1,2,
\end{eqnarray}
along with the condition that $\sum_{a,b}p(a,b|x,y)=1$ for all $x,y$.
This simple relation \eqref{SOS} is in fact sufficient to self-test the unknown state and measurements that result in the maximal violation of the steering inequality \eqref{Stefun1}. Let us now proceed to the main result of this work.  

\begin{thm}\label{Theo1M} 
Consider that the steering inequality \eqref{Stefun} is maximally violated by a state $\ket{\psi}_{AB}\in\mathbbm{C}^2\otimes\mathcal{H}_B$ and three-outcome measurements $B_y=\{N^b_y\} \ (y=0,1,2)$ acting on $\mathcal{H}_B$. Alice is trusted and her measurements $A_x$ are given in \eqref{Aideamea}. Then, 
there exist a local unitary transformation on Bob's side, $U_B$
such that
\begin{eqnarray}\label{Theo1.2M}
(\mathbbm{1}_A\otimes U_B)\ket{\psi}_{AB}=
\ket{\phi^+}_{AB},
\end{eqnarray}
and,
\begin{eqnarray}\label{Theo1.1M}
U_B\,\Pi_BN^b_y\Pi_B\,U_B^{\dagger}=\frac{2}{3}\proj{e^{*\perp}_{b,y}},
\end{eqnarray}
where $\ket{e_{b,y}}$ are given in \eqref{Aideamea} and $\Pi_B$ is the projector onto the support of $\rho_B=\Tr_A\left(\proj{\psi}_{AB}\right)$.
\end{thm}
\begin{proof}
We begin by considering a state $\rho_{AB}$ that maximally violates the steering inequality \eqref{Stefun1}. However, Bob's dimension is unrestricted and thus we can always purify this state by adding an auxiliary system to Bob. Thus without loss of generality, we consider that the state that results in the maximal violation is given by  $\ket{\psi}_{AB}$. 

Now as dimension of Alice's Hilbert space is $2$, as suggested in \cite{sarkar6} let us consider the Schmidt decomposition of the state $\ket{\psi}_{AB}$ as,
\begin{equation}\label{Schmidt}
    \ket{\psi}_{AB}=\sum_{i=0,1}\lambda_i\ket{s_i}_A\ket{t_i}_B,
\end{equation}
where the coefficients $\lambda_i> 0$ and satisfy the condition $\sum_{i}\lambda_i^2=1$.
The local vectors $\ket{s_i}\in\mathbbm{C}^2$ and $\ket{t_i}\in\mathcal{H}_B$ 
are orthonormal. Notice that the coefficients $\lambda_i\ne 0$ as any violation of the steering inequality \eqref{Stefun1} imposes that the state $\ket{\psi}_{AB}$ is entangled.

Let us now observe that there exists a unitary $U_B$ such that 
$U_B\ket{t_i}=\ket{s_i^*}$ for every $i$. Thus, the state \eqref{Schmidt} can be expressed as
\begin{eqnarray}\label{state2}
(\I_A\otimes U_B)\ket{\psi}_{AB}=(\I_A\otimes P_B)
\frac{1}{\sqrt{2}}\sum_{i=0,1}\ket{s_i}\ket{s_i^*},
\end{eqnarray}
where 
\begin{eqnarray}\label{PB}
P_B=\sqrt{2}\sum_{i=0,1}\lambda_i\proj{s_i^*}.
\end{eqnarray}
Notice that $P_B$ is full-rank as $\lambda_i\ne 0$ in the state \eqref{Schmidt}. The state on the right hand of Eq. \eqref{state2} is the two-qubit maximally entangled state. Thus,
\begin{equation}\label{PW}
(\I_A\otimes U_B)\ket{\psi}_{AB}=\ket{\Tilde{\psi}}_{AB}=(\I_A\otimes P_B)\ket{\phi^+}_{AB}.
\end{equation}

Let us now consider that Bob's measurements are POVM's given by $B_y=\{N^b_y\}$ such that $b,y=0,1,2$. We can characterise these measurements only on the support of Bob's reduced state $\rho_B$. Thus, we project these measurements onto the support of  $\rho_B$, to get
\begin{eqnarray}
\Pi_B N^b_{y}\Pi_B=\overline{N}^b_{y}
\end{eqnarray}
where $\Pi_B=\proj{t_0}+\proj{t_1}$ such that $\ket{t_i}$ are specified in Eq. \eqref{Schmidt}. Now, as shown in \cite{math}, a product of two positive semi-definite matrices is also positive semi-definite. Thus, $\overline{N}^b_{y}$ is also positive semi-definite as $\Pi_B, N^b_{y}$ are both hermitian and positive semi-definite matrices. The condition $\sum_{a,b}p(a,b|x,y)=\sum_{b}p(b|y)=1$, imposes that $\sum_b\overline{N}^b_y=\I_B$ for all $y$. Applying the unitary $U_B$, we arrive at
\begin{eqnarray}\label{meas}
U_B\overline{N}^b_{y}U_B^{\dagger}=\tilde{N}^b_{y}\qquad \forall b,y.
\end{eqnarray}

Notice from the above formula \eqref{meas} that $\tilde{N}^b_{y}$ acts on the Hilbert space $\mathbbm{C}^2$. 
Now, evaluating the joint probability $p(a,a|x,x)$ using the state \eqref{PW} and the measurements \eqref{meas}, we obtain
\begin{equation}
p(a,a|x,x)=\bra{\phi^+}(\I_A\otimes P_B)[M^a_x\otimes \tilde{N}^a_{x}](\I_A\otimes P_B)\ket{\phi^+},
\end{equation}
where $M^a_x$ denote Alice's measurement elements and are given in Eq. \eqref{Aideamea}. Now using the condition \eqref{SOS}, we arrive at 
\begin{eqnarray}\label{26}
p(a,a|x,x)=\bra{\phi^+}M^a_{x}\otimes P_B\tilde{N}^a_{x}P_B\ket{\phi^+}=0.
\end{eqnarray}
Using the fact that $R\otimes Q\ket{\phi^+}=\I\otimes QR^T\ket{\phi^+}$, where $R^T$ denotes the transpose of $R$ in the standard basis, we get from Eq. \eqref{26} that
\begin{eqnarray}\label{27}
\Tr\left[P_B\tilde{N}^a_{x}P_B.M^{aT}_{x}\right]=0.
\end{eqnarray}
Now, notice from \eqref{PB} that $P_B$ and $\tilde{N}^a_{x}$ are positive semi-definite. Thus $P_B\tilde{N}^a_{x}P_B$ is also positive semi-definite \cite{math}. 
Now, we take the eigendecomposition of $P_B\tilde{N}^a_{x}P_B$ as
\begin{eqnarray}\label{genPOVM}
P_B\tilde{N}^a_{x}P_B=\sum_{i=0,1}\alpha_{i,a,x}\proj{k_{i,a,x}}
\end{eqnarray}
such that $\alpha_{i,a,x}\geq0$. Expanding $M^a_x$ using \eqref{Aideamea}, we obtain from Eq. \eqref{27} that
\begin{eqnarray}\label{30}
\sum_{i=0,1}\alpha_{i,a,x}\left|\langle e_{a,x}^*\ket{k_{i,a,x}}\right|^2=0,
\end{eqnarray}
where we used the fact that for any projector $\Pi^T=\Pi^*$.
As $\tilde{N}^a_{x}$ acts on $\mathbbm{C}^2$ along with the fact that $\ket{e_{a,x}}\in \mathbbm{C}^2$ for any $a,x$, we expand $\ket{k_{i,a,x}}$ in the basis $\{\ket{e_{a,x}^*},\ket{e_{a,x}^{*\perp}}\}$ to obtain from \eqref{30} that
\begin{eqnarray}\label{31}
P_B\tilde{N}^a_{x}P_B=\beta_{a,x}\proj{e_{a,x}^{*\perp}},
\end{eqnarray}
where $\beta_{a,x}> 0$.
Then, using the fact that $\sum_a\tilde{N}^a_{x}=\I$ for any $x$ we get 
\begin{eqnarray}\label{33}
P_B^2=\sum_a\beta_{a,x}\proj{e_{a,x}^{*\perp}}\quad \forall x.
\end{eqnarray}
Thus, $\beta_{a,x}$ must satisfy the following condition
\begin{eqnarray}
\sum_a\beta_{a,x}\proj{e_{a,x}^{*\perp}}=\sum_a\beta_{a,x'}\proj{e_{a,x'}^{*\perp}}
\end{eqnarray}
for any $x,x'=0,1,2$.
Solving the above conditions by putting in the explicit form of $\ket{e_{a,x}}$ \eqref{Aideamea}, we get that $\beta_{a,x}=\beta_{a',x'}$ for any $a,x,a',x'$. Thus, from Eq. \eqref{33} we arrive at
\begin{eqnarray}
P_B^2=\frac{3\beta_{0,0}}{2}\I_B,
\end{eqnarray}
where we used the fact that $\sum_a\proj{e_{a,x}^{*\perp}}=3/2\ \I$ for any $x$.
This implies from \eqref{PW} that
\begin{eqnarray}
\ket{\Tilde{\psi}}_{AB}=\sqrt{\frac{3\beta_{0,0}}{2}}\ket{\phi^+}_{AB}.
\end{eqnarray}
Normalising the above state, we get that $\beta_{0,0}=2/3$. Thus, we have that the state upto some local unitary $U_B$ is the two-qubit maximally entangled state while the measurements from Eq. \eqref{31} is
\begin{eqnarray}
\tilde{N}^a_{x}=\frac{2}{3}\proj{e_{a,x}^{*\perp}}.
\end{eqnarray}
This completes the proof.
\end{proof}

{\it{Robustness against white noise.}} From an experimental perspective, it is important to find the robustness of our certification scheme against noise that might be present in the sources or the detectors. However, to perform an experiment it is not always necessary to find the full robustness, which captures the fidelity between the real and ideal state with respect to the violation of the steering inequality.  

Here, inspired by practical experiments, we find the robustness of our scheme with respect to a specific noise model, that is, when the ideal states and measurements, that result in the maximal violation of the steering inequality \eqref{Stefun1}, are mixed with white noise. For this purpose, let us consider ideal Bob's POVM's  $B_x=\{\frac{2}{3}\ket{e_{a,x}^{*\perp}}\!\bra{e_{a,x}^{*\perp}}\}_{a=0,1,2}$ where ${x=0,1,2},$ and $\ket{e_{a,x}}$ given in Eq. \eqref{Aideamea}. Adding white noise to every measurement element and defining the new measurement as $\overline{B}_x=\{N^a_x\}_{a=0,1,2}$ such that 
\begin{eqnarray}\label{robustBobmea}
N^a_x=\frac{2}{3}\left((1-\varepsilon_{a,x})\proj{e_{a,x}^{*\perp}}+\frac{\varepsilon_{a,x}}{2}\I\right).
\end{eqnarray}
Similarly, adding white noise to the maximally entangled state, we obtain that the noisy state shared between Alice and Bob is
\begin{eqnarray}\label{robuststate}
\rho_{AB}=(1-2\varepsilon_{s})\ket{\phi^+}\!\bra{\phi^+}_{AB}+\frac{\varepsilon_s}{2}\I.
\end{eqnarray}
It is worth noting here that the term $1-\varepsilon_i$ for any index $i$ is usually referred to as the visibility parameter.
Notice that the measurement elements and state being positive semi-definite, imposes that noise parameters $\varepsilon_{a,x}\geq0$ for any $a,x$ along with $\varepsilon_s\geq0$. Let us denote $\varepsilon=\max\{\max_{a,x}\{\varepsilon_{a,x}\},\varepsilon_s\}$. Without loss of generality, we can replace all the noise parameters $\varepsilon_{a,x},\varepsilon_s$ in Eqs. \eqref{robustBobmea} and \eqref{robuststate} with $\varepsilon$. 

Let us now evaluate the steering functional \eqref{Stefun}, with Alice being trusted and her measurements are given in \eqref{Aideamea}, using the noisy state \eqref{robuststate} and noisy Bob's measurements \eqref{robustBobmea}. For this purpose, let us first compute $p(0,0|0,0)$ as
\begin{eqnarray}
p(0,0|0,0)=\frac{2}{3}\Tr\left(\proj{e_{0,0}}\otimes N^0_{0}\rho_{AB}\right).
\end{eqnarray}
Substituting $\ket{e_{0,0}}$ from \eqref{Aideamea}, $N^0_{0}$ from \eqref{robustBobmea} and $\rho_{AB}$ from \eqref{robuststate}, we obtain that
\begin{eqnarray}
p(0,0|0,0)=\frac{\varepsilon}{9}(3-2\varepsilon).
\end{eqnarray}
Proceeding in a similar manner, we obtain for any $a,x=0,1,2$ that 
\begin{eqnarray}
p(a,a|x,x)=\frac{\varepsilon}{9}(3-2\varepsilon).
\end{eqnarray}
Thus, the value of the steering functional \eqref{Stefun}  when the ideal states and measurements are mixed with white noise is given by
\begin{eqnarray}
W=3+2\varepsilon^2-3\varepsilon\geq3(1-\varepsilon).
\end{eqnarray}
Thus, the proposed self-testing scheme is highly robust against white noise as the value of steering functional changes linearly with respect to the noise parameter $\varepsilon$. 

Let us also analyse the robustness of our protocol when the state shared between Alice and Bob has a noise model of the form
\begin{eqnarray}\label{robuststate1}
\rho_{AB}=(1-2\varepsilon_{s})\ket{\phi^{+,\delta}}\!\bra{\phi^{+,\delta}}_{AB}+\frac{\varepsilon_s}{2}\I
\end{eqnarray}
where
\begin{eqnarray}
   \ket{\phi^{+,\delta}}=\frac{1}{\sqrt{2(1+(1-\delta)^2)}} \left(\ket{00}+(1-\delta)\ket{11}\right).
\end{eqnarray}
Evaluating the steering functional $W$ \eqref{Stefun} using the above state \eqref{robuststate1} and noisy Bob's measurements \eqref{robustBobmea}, we get that
\begin{eqnarray}
    W=3-f(\delta,\varepsilon)
\end{eqnarray}
such that
\begin{equation}
  \small{  f(\delta,\varepsilon)=\frac{3 \sqrt{2} \delta (3 - 2 \varepsilon) \varepsilon + 
 3 \varepsilon (-3 + 2 \varepsilon) + \delta^2 (-2 + \varepsilon) (1 + 
    2 \varepsilon))}{-3 + 3 (\sqrt{2} - \delta) \delta}.}
\end{equation}
Thus, even when the state consists of noise along with an imbalance in the coefficient of the maximally entangled state, our scheme is highly robust as $f(\delta,\varepsilon)\backsim O(\delta,\epsilon)$ when $\varepsilon,|\delta|<<1$.



\textit{Conclusions---} There are a few certification schemes in the prepare-and-measure scenarios, for instance \cite{Armin2, Marcin}, that utilise non-projective measurements. However, none of these schemes can certify entangled states. In this work utilising the quantum steering scenario,
we propose a scheme for certification of the two-qubit maximally entangled state using non-projective measurements. Along with it, we also certified three three-outcome extremal POVM's on the untrusted Bob's side. We then show that our scheme is highly robust against the presence of white noise in the experimental devices. It is worth noting here that the certification of states using non-projective measurements can not be implemented in the standard Bell scenario. Thus, we identify a task that can be done using quantum steering but not using Bell nonlocality. This work also suggests that quantum steering might be useful towards designing highly noise-tolerant self-testing schemes, that is, quantum states might be certifiable using noisy projective measurements in the quantum steering scenario.

Some follow-up questions arise from our work. Firstly, it will be interesting to find 1SDI certification of any pure two-qubit entangled state using only non-projective measurements. A challenging problem in this regard would be to find a 1SDI scheme that can certify states of arbitrary dimension using only POVM's. In this work, we utilised extremal measurements, however, it will be interesting if one can find similar certification schemes using non-extremal measurements. 

\textit{Acknowledgement---} We would like to thank Remigiusz Augusiak for his valuable comments. This work is supported by Foundation for Polish Science through the First Team project (no First TEAM/2017-4/31).
\providecommand{\noopsort}[1]{}\providecommand{\singleletter}[1]{#1}%

\end{document}